\documentclass[11pt]{article}
\usepackage[paperwidth=8.5in,paperheight=11in,top=1in, bottom=1in, left=1in, right=1in]{geometry} 

\usepackage{amsfonts,amsthm,amssymb,amsmath}
\allowdisplaybreaks
\usepackage{xcolor}
\usepackage{tikz}
\usepackage{bm}

\usepackage{array}
\usepackage{mathtools}
\mathtoolsset{showonlyrefs=true}

\usepackage[numbers]{natbib}

\usepackage{hyperref} 
\hypersetup{
	colorlinks   = true,
	citecolor    = blue,
	linkcolor    = red,
	urlcolor     = magenta
}

\newtheorem{theorem}{Theorem}[section]

\newtheorem{proposition}[theorem]{Proposition}
\newtheorem{remark}[theorem]{Remark}

\newtheorem{lemma}[theorem]{Lemma}

\newtheorem{prob}[theorem]{Problem}
\newcommand{\Pb}{\mathbb{P}}

\newtheorem{corollary}[theorem]{Corollary}

\newcommand{\Eb}{\mathbb{E}}
\newcommand{\Fb}{\mathbb{F}}
\newcommand{\Rb}{\mathbb{R}}

\newcommand{\Ac}{\mathcal{A}}
\newcommand{\Fc}{\mathcal{F}}
\newcommand{\Jc}{\mathcal{J}}
\newcommand{\Lc}{\mathcal{L}}
\newcommand{\Pc}{\mathcal{P}}
\newcommand{\Tc}{\mathcal{T}}
\newcommand{\LP}{\mathcal{L}^2(\Pb)}
\newcommand{\LPL}{\mathcal{L}^2_{loc}(\Pb)}
\newcommand{\Uc}{\mathcal{U}}
\newcommand{\Vc}{\mathcal{V}}

\usepackage{bbm}
\newcommand{\In}{\mathbbm{1}}

\begin{document}

\title{\Large \scshape Quadratic  Hedging  for  Sequential Claims  \\with Random Weights in Discrete Time}

\author{
	Jun Deng\thanks{School of Banking and Finance,
		University of International Business and Economics, Beijing, China. Email: jundeng@uibe.edu.cn}  
	\and 
	Bin Zou\thanks{Corresponding author.   Department of Mathematics, University of Connecticut, 341 Mansfield Road U1009, Storrs Connecticut 06269-1009, USA. Email: bin.zou@uconn.edu. Phone: +1-860-486-3921.}
}

\date{First Version: May 12, 2020\\This Version: \today }
\maketitle

\vspace{-2ex}
\begin{abstract}
We study a quadratic hedging problem for a sequence of contingent claims with random weights in discrete time. 
We obtain the optimal hedging strategy explicitly in a recursive representation,  \emph{without} imposing the non-degeneracy (ND) condition on the model and square integrability on hedging strategies. 
We relate the general results to hedging under random horizon  and fair pricing in  the quadratic sense. 
We illustrate the significance of our results in an example in which the ND condition fails.
\end{abstract}



\section{Introduction}
\label{sec:intro}

Quadratic hedging, also called variance-optimal hedging, is an important topic in finance, with roots tracing back to the classical mean-variance criterion in modern portfolio theory. 
In a pioneering  work Schweizer \cite{schweizer1995variance}, the author seeks the optimal hedging strategy to  minimize the expected quadratic hedging error of a contingent claim in a discrete time model. 
In this paper, 
we generalize the work of \cite{schweizer1995variance} in two noticeable directions: (1) the hedging subject is no longer a single claim, but a \emph{sequence} of claims with random weights, and (2) the risky asset is \emph{not} required to satisfy the non-degeneracy (ND) condition (see its definition on p.4 in \cite{schweizer1995variance}). 

In standard quadratic hedging problems,
an investor constructs a portfolio from tradable assets 
to hedge the risk exposure of  a contingent claim $H_N$, which matures at time $N$. 
Let us consider a discrete time market consisting of one risk-free asset with zero interest rate  and one risky asset with price denoted by $S=(S_n)_{n=0,1,2,\cdots,N}$. 
Let $\xi = (\xi_n)_{n=1,2,\cdots,N}$ denote the investor's portfolio strategy, where $\xi_n$ is the number of shares invested in the risky asset from $n-1$ to $n$.
The investor aims to find an optimal hedging strategy $\xi^\star$ that minimizes the quadratic hedging error:
\begin{align}
	\label{eq:prob}
	V(c) := \min_{\xi \in \Ac} \; \Eb \big[(H_N - c - G_N(\xi))^2\big], 
\end{align}
where constant $c$ is the initial capital (or interpreted as the hedging cost), 
$G_N(\xi)  = \sum_{n=1}^N  \xi_n (S_n - S_{n-1})$ is the cumulative gain process at time $N$ under strategy $\xi$, and $\Ac$ is the set of admissible strategies.

Solving  Problem \eqref{eq:prob}  in a discrete time model dates back to \cite{schal1994quadratic} and \cite{schweizer1995variance}.  In \cite{schal1994quadratic, schweizer1995variance}, a necessary technical assumption is that
$(\Eb_n[\Delta S_{n+1}])^2/\Eb_n[(\Delta S_{n+1})^2]$ must be bounded by a constant less than 1. 
Here,  $\Delta S_{n+1} = S_{n+1} - S_n$  and  $\Eb_n[\cdot]$ denotes the conditional expectation operator (given $\Fc_n$).
Such an assumption is called the \emph{non-degeneracy} (ND) condition in  \cite{schweizer1995variance}.
With the ND condition in place, 
\cite{schweizer1995variance} presents a general analysis of Problem \eqref{eq:prob} and  obtains the optimal strategy  in a recursive form.
In both \cite{schal1994quadratic, schweizer1995variance} (and most subsequent works on this topic), 
the existence of a solution to Problem \eqref{eq:prob} is obtained using the Hilbert projection theorem, which requires the subspace of $\{G_N(\xi) : \xi \in \Ac \}$ to be closed and in turn needs the ND condition; while on the other hand, 
finding the optimal strategy is based on the Kunita–Watanabe decomposition. 
To overcome the restriction of the ND condition, \cite{melnikov1999mean} studies the conditional version of Problem \eqref{eq:prob}, replacing $\Eb[\cdot]$ by $\Eb[\cdot|\Fc_0]$ in the criterion, where $\Fc_0$ is the sigma field at time 0 and may be non-trivial. 
\cite{vcerny2004dynamic} applies dynamic programming to derive the solution of Problem \eqref{eq:prob}. 
\cite{vcerny2009hedging} utilizes the sequential regression method to find the optimal strategy.

The continuous version of Problem \eqref{eq:prob} is first studied by \cite{duffie1991mean}, in which $H_N$ is treated as a non-tradable asset and a hedger dynamically trades  another \emph{correlated} asset in a standard geometric Brownian motion model. 
\cite{schweizer1992mean} extends the work of \cite{duffie1991mean} by considering a general claim that may depend on both tradable and non-tradable assets. Further extensions are done in \cite{schweizer1996approximation, pham1998mean, vcerny2007structure} 
under a general semimartingale framework.  There is an extensive body of literature by now on this topic, and  we refer readers to the survey papers \cite{schweizer2001guided, schweizer2010mean} and the references therein for a detailed overview of pricing and hedging under a quadratic criterion in continuous time up to 2010.  Recent contributions to  quadratic hedging problems (in both discrete and continuous time) include approximations to expected utility  \cite{markowitz2014mean}, 
a mean-field formulation \cite{yi2014mean}, 
Monte-Carlo simulations \cite{cong2016multi}, 
trading constraint  \cite{cui2017mean},
empirical robust test \cite{gotoh2018robust},
and large portfolios   \cite{ao2019approaching}.

In this article, we extend the analysis of Problem \eqref{eq:prob} from a single claim $H_N$ to a sequence of claims $H = (H_n)_{n = 0,1,\cdots,N}$ with random weights $\omega = (\omega_n)_{n=0,1,\cdots,N}$, where both $H$ and $\omega$  are adapted to a given filtration and $\omega_n \in [0,1]$ for all $n$. 
Precisely, we consider the following quadratic hedging problem:
\begin{align}
	\label{eq:main}
	V(c) := \min_{\xi \in \Ac} \; \Jc(\xi; c) = \min_{\xi \in \Ac} \; \sum_{n=0}^{N} \, \Eb \left[\omega_n \big(H_n - c - G_n(\xi)\big)^2\right], \;\;
\end{align}
where $c \in \Rb$ is the initial capital (hedging cost) and $\Ac$ is the set of admissible strategies (precise definition given later).
It is clear that Problem \eqref{eq:prob} is a special case of our Problem \eqref{eq:main}, by simply taking $\omega_0 =\cdots = \omega_{N-1}=0$ and $\omega_N=1$.  The need of studying the extended Problem \eqref{eq:main} arises naturally from practical problems in insurance and finance, such as hedging under random horizon and minimizing index tracking error.  Further motivations of Problem \eqref{eq:main} are discussed in Section \ref{sec_interpretation_prob}.

Our paper contributes to the literature in several ways. 
First, we do not impose the ND condition on the price process $S$, 
which is required in most of the related works (see, e.g., \cite{schal1994quadratic, schweizer1995variance, vcerny2009hedging}).
In particular, in the example of Section \ref{secexam},  the ND condition fails, and consequently the standard results (as in \cite{schal1994quadratic, schweizer1995variance, vcerny2009hedging}) cannot be applied to solve Problem \eqref{eq:prob} or \eqref{eq:main}. 
However, our analysis still leads to a complete solution of Problem \eqref{eq:main}, i.e., we 
are able to obtain the optimal hedging strategy $\xi^\star$, the value function $V(c)$, and the minimum hedging cost $c^\star = \arg\min \, V(c)$ in closed-form. 
An important step to relax the ND condition is to consider a larger class of admissible strategies. 
To be precise, an admissible strategy $\xi$ in \cite{schal1994quadratic, schweizer1995variance} satisfies $\Eb[(\xi_n \Delta S_n)^2] < \infty$; but we do not require any square integrability condition on the strategy $\xi$. In fact, $\xi^\star \Delta S$ in our Example \ref{secexam} is only \emph{locally} square integrable.
Second, Problem \eqref{eq:main} considered in this paper is general enough to include the standard problem formulated in \eqref{eq:prob} as a special case and, to the best of our knowledge, has not been studied before.

We organize the rest of this paper as follows. In Section \ref{sec:prob}, we formulate the main problem. Section \ref{secmainresult} summarizes the main results. 
We discuss quadratic pricing in Section \ref{sec:pricing}.
In Section \ref{secexam}, we consider a concrete example in which the ND condition fails. 
Section \ref{secProofs} collects the proof to our main theorem.
In the online appendix, we provide technical proofs and one extra example with random horizon.

\section{The Problem}
\label{sec:prob}

We consider a discrete time market model with $N$ periods (where $N$ is a positive integer). For future convenience, we introduce two index sets of time by $\Tc = \{0,1,2,\cdots,N\}$ and $\Tc^+ = \{1,2,\cdots,N\}$.
Throughout the paper we fix a filtered probability space $(\Omega, {\cal F}, {\mathbb{F}}=({\cal F}_n)_{n \in \Tc}, \Pb)$, which supports all the random objects introduced in the sequel.  
$\Eb$ denotes taking  expectation under measure $\Pb$.
We consider a representative  investor who can trade a risk-free asset (bond) and a risky asset (stock) in this market. 
For convenience, we normalize the risk-free asset and set the interest rate to be zero. 
The price process of the risky asset is given by an $\mathbb{F}$-adapted and  \textit{locally
	square integrable} process $S:=(S_n)_{n \in \Tc}$ in the sense of \cite{vcerny2009hedging}, i.e.,  $\Eb[(\Delta S_n)^2 \vert {\cal F}_{n-1}]<+\infty$ for all $n  \in \Tc^+$. Here, we denote the price increment $\Delta S$ by $\Delta S_n = S_n - S_{n-1}$ (with $\Delta S_0 = 0$).
Introduce $\LP$ (resp. $\LPL$) as the set of all (resp. locally) square integrable random variables under $\Pb$. 
Denote $\Pc(\mathbb{F})$ the set of all $\mathbb{F}$-predictable processes, i.e., if $\psi=(\psi_n)_{n \in \Tc^+} \in \Pc(\Fb)$, we have $\psi_n \in \Fc_{n-1}$. We set $\psi_0 = 0$ for any predictable process $\psi$ unless stated otherwise.

A hedging strategy is a \emph{predictable} process $\xi = (\xi_n)_{n \in \Tc^+} \in \Pc(\Fb)$, where $\xi_n$ is the number of shares in the risky asset held by the investor from time $n-1$ to time $n$. 
The cumulative gain process of strategy $\xi$ is denoted by  $G(\xi) = (G_n(\xi))_{n \in \Tc}$.
A strategy $\xi$ is called \emph{self-financing} if 
\begin{align}
	\label{eq:G}
	G_n(\xi) = \sum_{i=1}^n \, \xi_i \, \Delta S_i, \quad n \in \Tc^+, \quad \text{ and } \quad 
	G_0(\xi) = 0.
\end{align}
A strategy $\xi$ is called  \emph{admissible} if it is predictable and self-financing. Denote the admissible set by $\Ac$. 

In our setup, the investor's risk exposure  is modeled by a \emph{sequence} of  contingent claims $H$ with random weights $\omega$.
Denote such a sequential risk by  $\Fb$-adapted processes $H = (H_n)_{n \in \Tc}$ and $\omega=(\omega_n)_{n \in \Tc}$, where claim $H_n  \in \LP$ and weight $\omega_n \in [0,1]$ for all $n$. 
We may still call $(H, \omega)$ a (contingent) claim for short.

\begin{prob}\label{our_problem}
	An investor, endowed with initial capital $c$, seeks an optimal hedging strategy that solves the following quadratic hedging problem for a sequence of randomly weighted claim $(H,\omega)$:
	\begin{align}
		\label{eq:main_1}
		V(c) := \min_{\xi \in \Ac} \; \Jc(\xi; c) = \min_{\xi \in \Ac} \; \sum_{n=0}^{N} \, \Eb \left[\omega_n \big(H_n - c - G_n(\xi)\big)^2\right], \;\;\;
	\end{align}
	where $G_n$ is given by \eqref{eq:G}.
	We call a solution $\xi^\star=\xi^\star(c)$ to Problem \eqref{eq:main_1} an optimal hedging strategy and $V(c)$ the value function or the minimum hedging error.
\end{prob}

Some remarks on the setup of Problem \eqref{eq:main_1} are in order.
\begin{remark}\label{remark_problem2}
	$(1)$ In this work, we do not impose the so-called ND condition on the price process $S$ of the risky asset, which is required in almost all existing works (see, e.g., \cite{schal1994quadratic, schweizer1992mean,  schweizer1995variance, schweizer1996approximation, vcerny2009hedging}).  The  example  of Section \ref{secexam} shows that  the optimal strategy $\xi^\star$ exists even though  $S$ there fails to satisfy the ND condition.
	\\
	$(2)$ \cite{melnikov1999mean} also considers a quadratic hedging problem without imposing the ND condition. There are several main differences between \cite{melnikov1999mean} and this paper. First, the hedging problem in \cite{melnikov1999mean} is to minimize $\Eb[(H_N - c - G_N(\xi))^2 | \Fc_0]$, a conditional version of Problem \eqref{eq:prob}, where $\Fc_0$ may \emph{not} be trivial. 
	Our main focus is Problem \eqref{eq:main_1}.
	Second, $S \in \Lc^2(\Pb)$ in \cite{melnikov1999mean}, but $S \in \Lc^2_{loc}(\Pb)$ in this paper.
	Third, the main results in \cite{melnikov1999mean} reply on a technical assumption that $V(c) < \infty$ (see Lemma 3 therein). In comparison, we prove $V(c) < \infty$ directly for Problem \eqref{eq:main_1} (see Part 2 of the proof to Theorem \ref{maintheorem1}).
	\\
	$(3)$ Problem \eqref{eq:main_1} in this paper is always well-posted, because the strategy of investing only in the risk-free asset (i.e., $\xi_n=0$ for all $n$) is feasible.
	\\
	$(4)$ Following \cite{vcerny2009hedging}, we only require $S$ to be locally square integrable (i.e., $S \in \LPL$). In comparison, the standard literature (e.g., \cite{schal1994quadratic,  schweizer1995variance}) assumes $S$ is square integrable (i.e.,  $S \in \LP$). 
	From the  example  of Section \ref{secexam} , we observe that $\xi^* \Delta S$ (and hence the gains process $G$) is not necessarily square integrable. 
	\\
	$(5)$ The assumption of random weights $\omega_n \in [0,1]$ can be relaxed, by scaling, to that $\omega_n$'s are bounded.
\end{remark}

\subsection{Applications of the Main Problem}
\label{sec_interpretation_prob}

In this subsection, we discuss several applications of Problem \eqref{eq:main_1} in finance and insurance.

First, we cast the framework of Problem \eqref{eq:main_1} for an index tracking problem, which is a popular form of passive fund management (see \cite{beasley2003evolutionary} and \cite{ruiz2009hybrid}). 
Here, let $H$ represent a \emph{non-tradable} market index that a fund manager wants to track by trading 
a related Exchange Traded Fund (ETF). For instance, to track the S\&P 500 index, the manager may consider Fidelity FNILX, Vanguard VOO, SPDR SPY, iShares IVV, or Schwab SWPPX. 
In \ref{sec_multiple}, we extend the analysis to multiple risky assets, and the results there allow the fund manager to trade multiple stocks and/or ETFs to track the target index.
The manager evaluates the tracking performance at \emph{multiple} time points up to some terminal time $N$, where performance is measured as the expected (weighted) square of tracking error. Often, the manager assigns equal weights to all points (i.e., $\omega_n = 1/N$); however, the setup in Problem \eqref{eq:main_1} also allows for more sophisticated and random weight assignment.

Second, the study of Problem \eqref{eq:main_1} can help insurance companies better price and hedge the risk exposure of certain equity-linked insurance policies. 
Here, we explain its application for a variable annuity (VA) policy with withdrawal guarantees (see \cite{dai2008guaranteed}). 
Denote by $G_n$ and $W_n$ respectively  the guaranteed amount upon withdrawal and the number of withdrawals at time $n$, then the risk exposure to the insurance company at time $n$ is given by $H_n:=W_n \cdot G_n$.
In practice, the guaranteed amount $G$ is often the minimum of a deterministic bound and an equity index. To hedge such VA contracts, the insurance company invests in the financial market with the purpose to minimize the aggregate of the expected quadratic hedging errors. 
The assignment of (random) weights $\omega_n$, $n \in \Tc$, may take into account the surrender (withdrawal) behavior of the policy holders. In the actuarial science literature (see \cite{milevsky2006financial} and \cite{bauer2008universal}), three common choices on the surrender behavior are: (1) no surrender, (2) 	static surrender, and (3) dynamic surrender. 
If there is no surrender, one may set $\omega_0 = \cdots = \omega_{N-1}=0$ and $\omega_N=1$. 
If surrender is static,  a known distribution can be used to model the surrender time and one may set the random weights as the probabilities of the chosen distribution.  
If surrender is dynamic, the surrender time $\tau$ is now a decision variable and one needs to solve an optimal stopping time first, which leads to the next formulation in \eqref{eq:ran_prob}.

Last, quadratic hedging problems under random horizon can be seen as a special case of Problem \eqref{eq:main_1}, although they can also be reduced to standard quadratic hedging problems.
To wit, let $\tau$ denote a random  exit time, a positive $\cal F$-measurable random variable taking values in $\Tc$. Consider a quadratic hedging problem with random horizon $\tau$ as follows:
\begin{align}
	\label{eq:ran_prob}
	\min_{\xi \in \Ac} \; \Eb\left( H_\tau - c - G_\tau(\xi) \right)^2, \quad H_n \in \LP \text{ for all } n \in \Tc.
\end{align}
It is easy to see that Problem \eqref{eq:ran_prob} is equivalent to Problem \eqref{eq:main_1} given ${ \omega_n = \Pb(\tau = n | \Fc_n)}$ for all $n \in \Tc$.	
At first glance, Problem \eqref{eq:ran_prob} may not have a solution since the stopped market $S^\tau$ may admit arbitrage opportunities (see, e.g., \cite{aksamit2017no, aksamit2018no, choulli2017no}). 
In the extra example in the online appendix, we solve Problem \eqref{eq:ran_prob} when the stopped market $S^\tau$ does have arbitrage opportunities. 
The formulation of Problem \eqref{eq:ran_prob} is of particular interest in insurance. 
For example, we may interpret $H_\tau$ as the payment of a life insurance contract, liquidated  at the random death time $\tau$, and consider an insurer who trades a longevity bond to hedge such a risk in discrete time. 

\section{Main Results}
\label{secmainresult}

In this section, we present the main results of this paper, a closed-form solution to Problem \eqref{eq:main_1}, in Theorem \ref{maintheorem1}. 
We first derive a sufficient optimality condition of Problem \eqref{eq:main_1} in the following proposition.

\begin{proposition}
	\label{prop:op}
	An admissible hedging strategy $\xi^\star$ is optimal to Problem \eqref{eq:main_1} if it satisfies, for all $i \in \Tc^+$, that  
	\begin{align}
		\label{eq:op_cond}
		\Eb    \left[ \left(\sum_{n=i}^{N} \omega_n (H_n - c - G_n(\xi^\star))\right)\Delta S_i\Big\vert {\cal F}_{i-1}\right] = 0.
	\end{align}
\end{proposition}

\begin{proof}
	Let $\xi \in \Ac$ be an arbitrary admissible strategy, and $\xi^\star \in \Ac$ satisfying condition \eqref{eq:op_cond}. 
	Noting $\xi + \xi^\star \in \Ac$ and the linearity of $G$ by \eqref{eq:G}, we obtain
	\begin{align}
		\Jc(\xi^\star + \xi; c) 
		&= \Jc(\xi^\star; c)  +  \sum_{n=0}^{N} \Eb \left[ \omega_n G_n( \xi)^2 \right]  \\
	 &\quad - 2 \sum_{i=1}^N \Eb \left[ \xi_i  \Eb    \left[ \left(\sum_{n=i}^{N} \omega_n (H_n - c - G_n(\xi^\star))\right)\Delta S_i \Big\vert {\cal F}_{i-1}\right] \right], 
	\end{align}
	which leads to the desired result.
\end{proof}


To facilitate the presentation of the main results, we define the following predictable processes  by 
\begin{align}
	\beta_n &:= \frac{\alpha_n}{\delta_n} \qquad \text{ and } \qquad   \rho_n :=   \frac{\eta_n}{\delta_n},  \quad \text{ for  all } n \in \Tc^+,
	\label{eq:beta}
\end{align}
where $\alpha = (\alpha_n)_{n \in \Tc^+}$, $\eta = (\eta_n)_{n \in \Tc^+}$, and $\delta = (\delta_n)_{n \in \Tc^+}$ are given by  
\begin{align}
	\alpha_n &:= \Eb\bigg[ \Delta S_n \bigg[\sum_{i=n}^{N} \omega_i \prod_{j=n+1}^{i} (1 - \beta_j \Delta S_j)\bigg] \bigg\vert {\cal F}_{n-1}\bigg],  
	\label{eq:alpha}
	\\
	\eta_n &:= \Eb\bigg[ \Delta S_n \bigg[\sum_{i=n}^{N} H_i \omega_i \prod_{j=n+1}^{i} (1 - \beta_j \Delta S_j)\bigg] \bigg\vert {\cal F}_{n-1}\bigg], 
	\label{eq:eta}
	\\
	\delta_n &:= \Eb \bigg[ \Delta S_n^2 \bigg[\sum_{i=n}^{N} \omega_i \prod_{j=n+1}^{i} (1 - \beta_j \Delta S_j)^2 \bigg] \bigg \vert {\cal F}_{n-1}\bigg].
	\label{eq:delta}
\end{align}
As convention,  we set a sum over an  empty  set to zero, a  product  over an  empty  set to one, and $0/0$ to zero.
The above definitions are made via backward induction. 
Namely, we first define, at time $n = N$, that 
$\alpha_N = \Eb[\omega_N \Delta S_N | \Fc_{N-1}]$, $\eta_N = \Eb[\omega_N H_N \Delta S_N | \Fc_{N-1}]$,
$\delta_N = \Eb[\omega_N \Delta S_N^2 | \Fc_{N-1}]$, $\beta_N = \frac{\alpha_N}{\delta_N}$, $\rho_N = \frac{\eta_N}{\delta_N}$,
and use induction to complete the definitions for all $n=N-1, N-2, \cdots, 1$.
We show that the processes in \eqref{eq:beta}-\eqref{eq:delta} are  well defined later in Corollary \ref{corowelldefined}.
By the definitions in \eqref{eq:beta}-\eqref{eq:delta}, these processes are predictable defined for $n \in \Tc^+$, and we further set them to 0 when $n=0$, i.e.,
\begin{align}
 \beta_0 = \rho_0 = \alpha_0 = \eta_0 = \delta_0 = 0. 
\end{align}

\begin{theorem}
	\label{maintheorem1} 
	For any fixed initial capital $c $, we define $\xi^\star(c) : = \left(\xi^\star_n(c)\right)_{n \in \Tc^+}$ by
	\begin{align}
		\label{eq:op_xi}
		\xi_n^\star(c) := \rho_n - \beta_n  \big( c + G_{n-1}(\xi^\star (c) ) \big),  
	\end{align}
	where $\rho=(\rho_n)_{n \in \Tc^+}$ and $\beta = (\beta_n)_{n \in \Tc^+}$ are given by \eqref{eq:beta} and $G$ is given by \eqref{eq:G}.
	The following results hold true:
	\begin{itemize}
		\item[\rm (a)] $V(c) = \Jc(\xi^\star; c) < \infty$ for all $c \in \Rb$, and the strategy $\xi^\star(c)$ defined in \eqref{eq:op_xi} solves Problem \eqref{eq:main_1}.
		
		\item[\rm (b)] The value function of Problem \eqref{eq:main_1} is given by 
		\begin{align}
			V(c) = & c^2  \, \sum_{n=0}^{N}  \Eb \left( Z_n \right) - 2c \sum_{n=0}^{N} \Eb \left(    Z_n H_n  \right) \\
			&+
			\sum_{n=0}^{N} \Eb \left( \omega_n \left(  H_n  - \sum_{i=1}^{n} \rho_i\Delta S_i \prod_{j=i+1}^{n} \left(1 - \beta_j \Delta S_j \right)    \right)^2 \right), \label{eq:value}
		\end{align}
		where $Z = (Z_n)_{n \in \Tc}$ is defined by 
		\begin{align}
			\label{eq:Z}
			Z_n := \omega_n \, \prod_{i=1}^n \, (1 - \beta_i \Delta S_i). 
		\end{align}
		
		\item[\rm (c)] Let $n$ be an arbitrary but fixed non-negative integer in $\Tc$, i.e., $n = 0,1,\cdots,N$.
		We have, for all $k = 0, 1, \cdots, n+1$, that  (setting $G_{-1} = 0$)
		\begin{align}
			H_n - c  - G_n(\xi^\star(c)) = H_n - \sum_{i=k}^n \rho_i \Delta S_i \prod_{j=i+1}^{n} (1 - \beta_j \Delta S_j) 
			- \big(c + G_{k-1}(\xi^\star(c)) \big) \prod_{i=k}^n (1 - \beta_i \Delta S_i).\qquad 
					\label{eq:H_thm}
		\end{align}
	\end{itemize}
\end{theorem}

\begin{corollary}
	If  the random weights process $\omega = (\omega_n)_{n \in \Tc}$ degenerates into a sequence of constants and $S$ is an $\mathbb{F}$-martingale, we have 
	\begin{align}
		\xi^\star_n = \rho_n = \frac{\sum_{i=n}^{N} \; \omega_i \, \Eb\left[ H_i \Delta S_n   \Big\vert {\cal F}_{n-1}\right]}{\sum_{i=n}^{N}  \; \omega_i  \, \Eb \left[ \Delta S_n^2     \Big\vert {\cal F}_{n-1}\right]}, \qquad n \in \Tc^+,
	\end{align}
	which is independent of the initial capital $c$. Note that if the denominator in the above ratio is zero, the numerator must also be zero. In such a scenario, we set  0/0 to zero by convention.
\end{corollary}

\begin{remark}
	\label{rem:multi}
In the model setup, we consider a financial market with only one risky asset. However, we note that such an assumption does not impose extra restrictions, as all the results in Theorem \ref{maintheorem1} hold under the corresponding vector notations. Please refer to \ref{sec_multiple} for details. 
In particular, such a positive extension result may prove to be valuable if a fund manager wants to trade multiple stocks to track an index, as mentioned in Section \ref{sec_interpretation_prob}.
\end{remark}

\section{Connection with Pricing}
\label{sec:pricing}

In this section, we explore the connection between hedging and pricing of a sequence of contingent claims in the quadratic sense. 
Key results are presented in Theorem \ref{thm:pricing}. 

Note from Theorem \ref{maintheorem1} that the optimal strategy $\xi^\star = \xi^\star(c)$ depends on the initial capital $c$.
If the minimum hedging error $V(c) = 0$ for some $c$ in \eqref{eq:main_1}, then the portfolio $(c, \xi^\star(c))$ replicates the claim $(H, \omega)$, and hence $c$ is the fair price of $(H, \omega)$ at time 0.
However, in a general incomplete market, it is most likely $V(c) > 0$ for all $c$ and one may further consider $\min_{c \in \Rb} \, V(c)$ to find the ``optimal" hedging cost $c^\star$, along with the corresponding optimal strategy $\xi^\star(c^\star)$. 
In other words, an investor, with objective given by \eqref{eq:main_1}, prefers the portfolio $(c^\star, \xi^\star(c^\star))$  over any admissible portfolio $(c, \xi)$.
Based on the above discussions, we formulate the quadratic pricing problem of a sequential claim $(H, \omega)$ as follows:
\begin{align}
	\label{pricingproblem}
	V^\star = \min_{c \in \Rb} \; V(c) = \min_{c \in \Rb}  \min_{\xi \in \Ac}  \sum_{n=0}^{N} \Eb \left[\omega_n (H_n - c - G_n(\xi))^2\right]. \quad
\end{align}

\begin{theorem}
	\label{thm:pricing}
	Problem \eqref{pricingproblem} has an optimal solution $(c^\star, \xi^\star(c^\star))$  given by 
	\begin{align}
		\label{eq:c_op}
		c^\star = \dfrac{\sum_{n=0}^{N}  \Eb \left(Z_n H_n \right)}{\sum_{n=0}^N \Eb(Z_n)} \qquad \text{and} \qquad 
		\xi^\star(c^\star) \text{ by } \eqref{eq:op_xi}, 
	\end{align}
	where $Z=(Z_n)_{n \in \Tc}$ is defined in \eqref{eq:Z}. 
\end{theorem}

\begin{proof}
	We have $\sum_{n=0}^N \Eb(Z_n)  \ge 0$ from \eqref{eq:Z_pos}, and when the equality holds, we set $c^\star=0$ (recall $0/0=0$ by convention). The rest is obvious from Assertion (b) in Theorem \ref{maintheorem1}.
	%
	In fact, we have $V^\star = \Jc(\xi^\star(c^\star); c^\star) \le \Jc(\xi^\star(c); c) \le \Jc(\xi; c)$ for all $c \in \Rb$ and $\xi \in \Ac$.  This ends the proof.
\end{proof}

Similar to Problem \eqref{eq:ran_prob}, we can reformulate the quadratic pricing problem in \eqref{pricingproblem} under a random horizon $\tau$, and apply Theorem \ref{thm:pricing} to obtain the solution to such a problem. The results are given below.

\begin{corollary}\label{coroforqudrarandomhorizon}
	Let $\tau$ be a random time and $(H, \omega)$ be a contingent claim, with $\omega_n = \Pb(\tau = n | \Fc_n)$. 
	The minimum capital $c^\star$, given in \eqref{eq:c_op},  solves the following pricing problem under random horizon:
	\begin{align}
		\min_{c \in \Rb} \; \min_{\xi \in \Ac} \;  \Eb  \left[ \big( H_\tau - c - G_\tau(\xi) \big)^2 \right].
	\end{align}
	
	Furthermore, if $\tau$ is independent of  $S$ and $S$ is an $\mathbb{F}$-martingale, the minimum capital $c^\star$ is equal to  
	$c^\star = \sum_{n=1}^{N}  \Eb\left(H_n \right)  \cdot \Pb (\tau =n) .$
\end{corollary}
\begin{remark}
	Denote $\widetilde{Z} = (\widetilde{Z}_n)_{n \in \Tc}$, where $\widetilde Z_n = Z_n / \sum_{i=0}^N \, \Eb[Z_i]$. 
	Then $c^\star$ given in \eqref{eq:c_op} can be rewritten as 
	$c^\star = \sum_{n=0}^N \, \Eb[\widetilde{Z}_n H_n]$. 
	Furthermore, if $\tau$ degenerates to a constant $N$, Corollary \ref{coroforqudrarandomhorizon} is reduced to   Corollary 3.2 in \cite{schweizer1995variance} and $\widetilde{Z}$ is reduced to a signed  probability measure that is absolutely continuous with respect to $\Pb$. 
	
	As pointed out in \cite{schweizer1995variance} (see Section 5 therein), $\widetilde{Z}$ may become negative with positive probability in certain models. 
	As a result, it is possible that $H_n \ge 0$ for all $n \in \Tc$ but $c^\star <0$. For this reason, Schweizer comments in \cite{schweizer1995variance}(p.29) that ``This (example) shows that an interpretation of either $H_0$ or $V_0$ (corresponding to our $c^\star$) as a fair price of $H$ does not always make sense from an economic point of view''.
	
\end{remark}

\section{An Example}\label{secexam}

In this section, we revisit the  Schachermayer's example in \cite{schweizer1995variance}, 
where the risky asset $S$ does \emph{not} satisfy the ND condition. 
Such an example is used in \cite{schweizer1995variance} to show that a \emph{square integrable} optimal hedging strategy to Problem \eqref{eq:prob} may fail to exist if the ND condition fails. 
However, the main theorem,  Theorem \ref{maintheorem1}, in this paper applies even when the ND condition fails. In addition, we mention that we do not impose square integrability on admissible strategies, which is essential to obtain an optimal strategy in this example (see Remark \ref{rem:exm} below).

Let $\Omega = [0,1]\otimes \{-1,1\}$ with its  Borel $\sigma$-algebra ${\cal F}$.  The element in $\Omega$ is denoted by 
$x=(u,v)$ with  $u\in [0,1]$ and $v\in \{-1,1\}$. Define two random variables $\Uc(x)  = u$ and $\Vc(x) = v$. Let $\mathbb{P}$ be a probability measure on  $(\Omega, {\cal F})$ such that $\Uc$ is uniform distributed on $[0,1]$ and 
\begin{align}
	\mathbb{P}(\Vc = 1 \vert \Uc)= \Uc^2 \quad \text{ and } \quad \mathbb{P}(\Vc = -1 \vert \Uc) = 1-\Uc^2.
\end{align}
The filtration is specified as 
${\cal F}_0 = {\cal F}_1 = \sigma(\Uc)$ and   ${\cal F}_2 = {\cal F}$.
The dynamics of the risky asset $S$ is given by  $S_0 = 0$,  $\Delta S_1 = 1$, and $\Delta S_2 = \Vc^+ \, (1+ \Uc)-1$. Here, $\Vc^+:=\max(\Vc,0)$ stands for the positive part of $\Vc$.  Following \cite{schweizer1995variance}, it is easy to see that
\begin{align}
	\frac{\left(\mathbb{E} [\Delta S_2 \vert {\cal F}_1]\right)^2}{\mathbb{E}[\Delta S_2^2 \vert {\cal F}_1]} = 
	\frac{\left(\Uc^3+\Uc^2-1\right)^2}{\Uc^4 - \Uc^2 +1}
\end{align}
is not uniformly bounded by a constant less than one. Therefore, the  price process $S$ does not satisfy the ND-condition, and the main theorem in \cite{schweizer1995variance} cannot be used to derive the optimal strategy in this market.

As noted previously, our main results in Theorem \ref{maintheorem1} apply even when the ND condition fails. 
To proceed, let us consider Problem \eqref{eq:main_1} with the claim $(H, \omega)$ given by
\begin{align}
	H_0 &=0, & H_1 &= 1, & H_2 &=\left(\frac{1}{\Uc}+1\right) \Vc^+; \\
	\omega_0 &= 0, & \omega_1 &= 1, & \omega_2 &= \Vc^+. 	\label{eq:cc_ex1}
\end{align}

Using \eqref{eq:alpha}-\eqref{eq:delta}, we first compute their values at time $2$: 
$\alpha_2 = \Eb[\omega_2 \Delta S_2 | \Fc_1] = \Uc^3$, $\eta_2 = \Eb[\omega_2 H_2 \Delta S_2 | \Fc_1] = \Uc^3 + \Uc^2$, and $\delta_2 =\Eb[\omega_2 \Delta S_2^2 | \Fc_1] = \Uc^4$,  
which imply
\begin{align}
	\beta_2 =  \frac{\alpha_2}{\delta_2} = \frac{1}{\Uc} \qquad \text{ and } \qquad 
	\rho_2 = \frac{\eta_2}{\delta_2} = \frac{1+\Uc}{\Uc^2}.  \label{eq:pa2_ex1}
\end{align}
At time 1, notice that $\omega_2 ( 1 - \beta_2 \Delta S_2) \equiv 0$. We then obtain 
$\alpha_1 = \Eb \left[\omega_1 \Delta S_1  \vert \Fc_0 \right] = 1$,  $\eta_1 = \Eb[\omega_1 H_1 \Delta S_1\vert \Fc_0  ] =1$, and  $\delta_1 =  \Eb \left[\omega_1 \Delta S_1^2 \vert \Fc_0 \right] = 1 $, leading to 
\begin{align}
	\beta_1 =\frac{\alpha_1}{\delta_1}=1 \qquad \text{and} \qquad  \rho_1 = \frac{\eta_1}{\delta_1} = 1. \label{eq:pa1_ex1}
\end{align}
Using the optimal strategy \eqref{eq:op_xi} in Theorem \ref{maintheorem1}, we have: 
\begin{corollary}\label{coro_ex1}
	Let $(H, \omega)$ be given by \eqref{eq:cc_ex1}.  For any  initial capital $c$, the optimal quadratic hedging strategy $\xi^\star=(\xi_1^\star, \xi_2^\star)$ to Problem \eqref{eq:main_1} is given by $\xi_1^\star = 1 -  c $ and $\xi_2^\star = \frac{1}{\Uc^2}$. 
\end{corollary}
Below, we consider a new claim $(\overline H,  \overline \omega)$, where $ \overline H = H$ in \eqref{eq:cc_ex1} but $ \overline \omega \neq \omega$ in \eqref{eq:cc_ex1}, given by 
\begin{align}
	\overline{H}_0 &=0, & \overline{H}_1 &= 1, & \overline{H}_2 &=\left(\frac{1}{\Uc}+1\right) \Vc^+;
	\\  
	\overline{\omega}_0 &= 0, & \overline{\omega}_1 &= 1, & \overline{\omega}_2 &= \Vc^-. 
		\label{eq:cc_ex1_extention}
\end{align}
It is straightforward to calculate that 
$\overline{\alpha}_2 = \Uc^2 - 1$,  $\overline{\eta}_2 = 0$, $\overline{\delta}_2 = 1- \Uc^2$,  
and  $\overline{\alpha}_1 = \overline{\eta}_1=\overline{\delta}_1 =1$,
which leads to $\overline{\beta}_2 = -1, \overline{\rho}_2=0, \overline{\beta}_1 =1, \overline{\rho}_1 =1$. Therefore,     the optimal strategy is  
$\overline{\xi}^\star_1 = 1-c$ and $\overline{\xi}^\star_2 = c$. 

\begin{remark}
\label{rem:exm}
	$\rm (1)$ The optimal hedging strategy $\xi^\star$ in Corollary \ref{coro_ex1} is  predictable and self-financing, and is then admissible in our definition. However, it is easy to verify that 
	\begin{align}
		\xi^\star_1  \Delta S_1 \in \LP,\;  \xi^\star_2  \Delta S_2 \in \LPL, \text{ and }   \xi^\star_2  \Delta S_2 \notin \LP,
	\end{align}
	which implies $\xi^\star$ in Corollary \ref{coro_ex1} is not admissible in the definition of \cite{schweizer1995variance}. In other words, we remove the ND condition in \cite{schweizer1995variance} by considering a larger class of admissible strategies than that in \cite{schweizer1995variance}.
	\\
	$\rm(2)$ For the claim $(\overline{H}, \overline{\omega})$ in \eqref{eq:cc_ex1_extention}, obviously,  $\overline{\xi}^\star_1  \Delta S_1 \in \Lc^2(\Pb)$   and $ \overline{\xi}^\star_2  \Delta S_2 \in \Lc^2(\Pb)$. 
	This result shows that, even for the same market model, whether the optimal strategy is square integrable depends on the sequential claims with random weights $(H, \omega)$.
	\\
	$\rm (3)$ The example in Corollary \ref{coro_ex1} sheds light on  Remark \ref{remark_problem2} that  we do not impose the ND condition on the risky price process $S$ and the square integrability condition on hedging strategies.
\end{remark}

\section{Proof to Theorem \ref{maintheorem1}}
\label{secProofs}
In this section, we provide the proof to Theorem \ref{maintheorem1}. 
To that end, we first present several preliminary results. 
We define processes $A=(A_n)_{n \in \Tc}$, $B=(B_n)_{n \in \Tc}$, $C=(C_n)_{n \in \Tc}$, and $D=(D_n)_{n \in \Tc}$ by 
\begin{align}
	A_n &:= \sum_{i=n}^{N} \omega_i \prod_{j=n+1}^{i} (1 - \beta_j \Delta S_j), \quad
	B_n:= A_n \Delta S_n, \\   
	C_n &:= \beta_n B_n, \quad
	D_n := \sum_{i=n}^{N} \omega_i \prod_{j=n+1}^{i} (1 - \beta_j \Delta S_j)^2, 	\label{eq:ABCD}
\end{align}
where $\beta$  is defined in \eqref{eq:beta}. 
By \eqref{eq:ABCD} and the definition of $\alpha=(\alpha_n)_{n \in \Tc^+}$ in \eqref{eq:eta}, we easily deduce, for $n=0,1,\cdots,N-1$, that 
\begin{align}
	\label{eq:A_eq}
\Eb[A_n | \Fc_n] = \Eb[A_{n+1} | \Fc_n] + \omega_n  - \alpha_{n+1}\beta_{n+1},
\end{align}
and $	A_N = \omega_N $.

\begin{lemma}
	\label{propintegrable1}
	Let processes $A$, $B$, $C$, and $D$ be defined by \eqref{eq:ABCD}. We have: $\rm (1)$  $A$ and $C$  are  square integrable ($A, C \in \LP$) and $B$ is \textit{locally} square integrable ($B \in \LPL$), and $\rm (2) $
	\begin{align}\label{AKDKleq}
		\Eb [A_n | \Fc_n ] = \Eb[D_n | \Fc_n]\leq  N-n+1, \qquad \forall \, n \in \Tc.
	\end{align}
\end{lemma}

\begin{proof}
Assertion (2)  can be shown by backward induction, and using \eqref{eq:ABCD} and \eqref{eq:A_eq}.
Assertion (1) is proved by using the Cauchy-Schwarz inequality, $\omega_n \in [0,1]$, and \eqref{AKDKleq}. 
Please see Section \ref{proof:ABCD} in the online appendix for the complete proof. 
\end{proof}

The following is an immediate application of Lemma \ref{propintegrable1}.

\begin{corollary}
	\label{corowelldefined}
	The processes $\beta$, $\rho$, $\alpha$, $\eta$, and $\delta$ given in \eqref{eq:beta}-\eqref{eq:delta} are well defined.
\end{corollary}

\begin{lemma}
	\label{lem:rho_beta}
	We have, for all $n \in \Tc^+$, that:  
	\begin{align}
	&  \Eb \left[ \Delta S_n \left(\sum\limits_{i=n}^{N}   \omega_i \left(  \sum\limits_{j=n+1}^{i} \rho_j\Delta S_j \prod\limits_{k=j+1}^{i} \left(1 - \beta_k \Delta S_k \right)\right) \right) \Bigg \vert {\cal F}_{n-1}\right]  \\
	= \; &\Eb \left[ \Delta S_n \left(\sum\limits_{i=n}^{N}   \omega_i H_i \left(  \sum\limits_{j=n+1}^{i} \beta_j\Delta S_j \prod\limits_{k=j+1}^{i} \left(1 - \beta_k \Delta S_k \right)\right) \right) \Bigg\vert {\cal F}_{n-1}\right].
	\end{align}
\end{lemma}

\begin{proof}
The desired equality is proved by noticing $\beta_n \eta_n = \alpha_n \rho_n$ from \eqref{eq:beta}. Please see Section \ref{proof:rho_beta} in the online appendix for the complete proof.
\end{proof}

We are now ready to show the three assertions in Theorem \ref{maintheorem1} and complete this task in four parts.

\begin{proof}[Proof to Theorem \ref{maintheorem1}] 

	\textbf{Part 1:} We first show \eqref{eq:H_thm} in Assertion (c) holds by backward induction. 
	We fix an integer $n \in \Tc$.
	When $k = n + 1$, \eqref{eq:H_thm} is trivial. 
	Next suppose \eqref{eq:H_thm} holds for all $k=n+1, n, n-1, \cdots,l+1$. 
	Applying the assumption for $k=l+1$, and using $G_l(\xi^\star) = G_{l-1}(\xi^\star) + \xi^\star_l \, \Delta S_l$ and the expression of $\xi^\star$ from \eqref{eq:op_xi}, we verify that \eqref{eq:H_thm} also holds for $k=l$.
	Assertion (c) is now proved.
	\\[2ex]
	\textbf{Part 2:} We show 
	$V(c) = \Jc(\xi^\star;c) < \infty$ for any $c \in \Rb$, which is done by checking $\Eb [\omega_n (H_n - c - G_n(\xi^\star))^2] < \infty$ for all $n \in \Tc$.
	To achieve this purpose, we obtain 
	{\small 
		\begin{align}
		\Eb \left(\omega_n \left(H_n - c - G_n(\xi^\star )\right)^2\right) 
			&= \Eb \Big( \omega_n \Big( H_n  - \sum_{i=0}^{n} \rho_i \Delta S_i \prod_{j=i+1}^{n} \left(1 - \beta_j \Delta S_j \right)   -   c  \prod_{i=0}^{n} \left(1 - \beta_i \Delta S_i\right) \Big)^2\Big)  \\
			&  \tag*{\text{(take $k=0$ in \eqref{eq:H_thm})}}\\
			&\leq 3  \, \Eb \Big(\omega_n  H_n^2 + \omega_n \Big(  \sum_{i=0}^{n} \rho_i \Delta S_i \prod_{j=i+1}^{n} \left(1 - \beta_j \Delta S_j \right)    \Big)^2 \\
			&\quad + \omega_n  c^2 \prod_{i=0}^{n} \left(1 - \beta_i \Delta S_i \right)^2 \Big) 
			 \tag*{\text{(by Cauchy-Schwatz)}} \\
			&\leq 3 \, \Eb \left( H_n^2 \right)  + 3  \, \Eb \Big( \omega_n \Big(  \sum_{i=0}^{n} \rho_i \Delta S_i \prod_{j=i+1}^{n} \left(1 - \beta_j \Delta S_j \right)    \Big)^2 \Big)  + 3c^2 \, \Eb (D_0)  \\
			& \tag*{($\omega_n \in [0,1]$ \text{ and } \eqref{eq:ABCD})} \\
			&\leq 3 \, \Eb \left( H_n^2 \right)  + 3 (n+1) \sum_{i=0}^{n} \Eb \Big( \rho_i^2 \Delta S_i^2   \sum_{j=i}^{N}  \omega_j  \prod_{l=i+1}^{j} \left(1 - \beta_l \Delta S_l \right)^2      \Big) \\
			&\quad + 3c^2 (N+1)   \tag*{\text{(take $\Eb$ for $D_0$ in \eqref{AKDKleq})}}\\ 
			&= 3\,  \Eb \left( H_n^2 \right)  + 3c^2 (N+1)+ 3 (n+1) \sum_{i=0}^{n} \Eb \left( \frac{\eta_i^2}{\delta_i}    \right) 
			 \tag*{\text{(use \eqref{eq:beta} and \eqref{eq:delta})}}  \\
			&\le  3\,  \Eb \left( H_n^2 \right)  + 3c^2 (N+1) 
			 + 3(n+1) \sum_{i=0}^{n} \Eb \Bigg[ \frac{1}{\delta_i} \, \Eb \Bigg[ \sum_{j=i}^{N} H_j^2  \Big\vert {\cal F}_{i-1}\Bigg] \\
			& \qquad \cdot \underbrace{\Eb \Bigg[\Delta S_i^2  \sum_{j=i}^{N} \omega_j \prod_{k=i+1}^{j} (1 - \beta_k \Delta S_k)^2  \Big\vert {\cal F}_{i-1}\Bigg]}_{=\delta_i} \Bigg]  
			  \tag*{(\text{use }\eqref{eq:eta} \text{ and  H\"older})} \\
			&= 3\,  \Eb \left( H_n^2 \right)  + 3c^2 (N+1)  + 3(n+1) \sum_{i=0}^{n}  \sum_{j=i}^N \, \Eb[H_j^2] <+\infty. 
		\end{align}
	} 
	In particular, we obtain $V(c) = \Jc(\xi^\star;c) < \infty$ without imposing the ND condition and $\xi^\star \in \Lc^2(\Pb)$.
	\\[2ex]
	\textbf{Part 3:} We show $\xi^\star = \xi^\star(c)$ given by \eqref{eq:op_xi} satisfies the sufficient condition \eqref{eq:op_cond} in Proposition \ref{prop:op}, and hence is optimal to Problem \eqref{eq:main_1}. 
	The proof below is based on backward induction. 
	
	When $n = N$, by using \eqref{eq:G}, we have 
	$	\Eb \big[\omega_N (H_N - c - G_N(\xi^\star)) \cdot \Delta S_N \big| \Fc_{N-1}\big] 
		=  \Eb \big[\omega_N H_N \Delta S_N \big| \Fc_{N-1}\big]  - \xi_N^\star \, \Eb \big[\omega_N  \Delta S_N^2 \big| \Fc_{N-1}\big]  - (c + G_{N-1}(\xi^\star)) \, \Eb \big[\omega_N  \Delta S_N \big| \Fc_{N-1}\big] 
		= \eta_N -  \delta_N \, \xi_N^\star - \alpha_N \, (c + G_{N-1}(\xi^\star))$, 
	which vanishes with $\xi_N^\star = \rho_N - \beta_N \, (c + G_{N-1}(\xi^\star)) $, where $\rho_N = \eta_N / \delta_N$ and $\beta_N = \alpha_N / \delta_N$ by \eqref{eq:beta}. 
	
	Next suppose the desired statement is true for $n = N, N-1, \cdots, k+1$. We aim to prove the same statement holds for $n=k$ as well. We first recall a useful identity  (which can be proven by induction)
	\begin{align}\label{eqINDENTITYone}
		\prod_{i=k+1}^{l} (1 - a_i) =  1- \sum_{i=k+1}^{l} a_i \prod_{j=i+1}^{l} (1 - a_j),
	\end{align}
	where $k$ and $l$ are fixed integers, and $a = (a_n)$ is any sequence. 
	We then obtain 
	{\small 
		\begin{align}
			&\Eb \Bigg[ \Bigg(\sum_{n=k}^N \omega_n (H_n - c - G_n(\xi^\star)) \Bigg) \cdot \Delta S_k \Bigg| \Fc_{k-1}\Bigg] \\
			= \; &  \Eb \left[ \left(\sum_{n=k}^{N}   \omega_n \left(H_n  - \sum_{i=k+1}^{n} \rho_i \Delta S_i \prod_{j=i+1}^{n} \left(1 - \beta_j \Delta S_j \right)\right) \right)\Delta S_k \Bigg\vert {\cal F}_{k-1}\right] \\
			& - \Eb \left[ \left(\sum_{n=k}^{N} \omega_n  \left(c + G_{k}(\xi^{\star})\right) \prod_{j=k+1}^{n} \left(1 - \beta_j \Delta S_j\right) \right)\Delta S_k \Bigg\vert {\cal F}_{k-1}\right]  \tag*{(\text{by } \eqref{eq:H_thm})} \\
			= \; & \Eb \Bigg[ \Delta S_k\left(  \sum_{n=k}^{N}   \omega_n  H_n \right)     \Bigg \vert {\cal F}_{k-1}\Bigg] \\
			& - \Eb \Bigg[ \Delta S_k  \sum_{n=k}^{N}   \omega_n \Big(  \sum_{i=k+1}^{n} \rho_i \Delta S_i \prod_{j=i+1}^{n} \left(1 - \beta_j \Delta S_j \right)\Big)   \Bigg \vert {\cal F}_{k-1}\Bigg]  \\ 
			&
			-\left(c + G_{k-1}(\xi^\star)\right) \Eb\Bigg[ \Bigg(\sum_{n=k}^{N} \omega_n   \prod_{j=k+1}^{n} \left(1 - \beta_j \Delta S_j \right) \Bigg)\Delta S_k \Bigg\vert {\cal F}_{k-1}\Bigg]   
			\\& 
			- \xi_k^\star \, \Eb \Bigg[ \Bigg(\sum_{n=k}^{N} \omega_n    \prod_{j=k+1}^{n} \left(1 - \beta_j \Delta S_j\right) \Bigg) \Delta S_k^2 \Bigg \vert {\cal F}_{k-1}\Bigg]  
			 \tag*{(\text{by } \eqref{eq:G})} 
			\\
			= \; & - \Eb\Bigg[ \Delta S_k \Bigg(\sum_{n=k}^{N}   \omega_n H_n \Bigg(  \sum_{i=k+1}^{n} \beta_i \Delta S_i \prod_{j=i+1}^{n} \left(1 - \beta_j \Delta S_j \right)\Bigg) \Bigg) \Bigg \vert {\cal F}_{k-1}\Bigg]\\
			& +  \Eb \Bigg[ \Delta S_k\left(  \sum_{n=k}^{N}   \omega_n  H_n \right)     \Bigg \vert {\cal F}_{k-1}\Bigg]  \tag*{(\text{Lemma \ref{lem:rho_beta}})} \\ 
			& - \alpha_k  \left(c + G_{k-1}(\xi^\star)\right) - \delta_k  \xi_k^\star   \tag*{(\text{by } \eqref{eq:eta}, \eqref{eq:A_eq}, \eqref{eq:delta})} \\
			= \; & \Eb \Bigg[ \Delta S_k \Bigg[\sum_{n=k}^{N}   \omega_n H_n  \prod_{i=k+1}^{n} \left(1 - \beta_i \Delta S_i \right)\Bigg]  \Bigg \vert {\cal F}_{k-1}\Bigg] 
			 - \alpha_k  \left(c + G_{k-1}(\xi^\star)\right) - \delta_k  \xi_k^\star
			 \tag*{\text{(by \eqref{eqINDENTITYone})}}\\
			= \; & \eta_k - \alpha_k  \left(c + G_{k-1}(\xi^\star)\right) - \delta_k  \xi_k^\star = 0,   \tag*{\text{(by \eqref{eq:eta} and \eqref{eq:op_xi})}}
		\end{align}
	}
	which confirms the induction indeed holds for $n=k$. 
	
	By the definition in \eqref{eq:op_xi}, $\xi^\star$ is predictable and self-financing, and hence solves Problem \eqref{eq:main_1}.
	\\[2ex]
	\textbf{Part 4:} We show that the value function $V(c)$ is given by \eqref{eq:value}.  
	Taking $k=1$ in \eqref{eq:H_thm} for all $n \in \Tc$, we get 
	{\small 
		\begin{align}
			\sum_{n=0}^{N} \Eb \big( \omega_n \left(H_n - c - G_n(\xi^\star) \right)^2 \big) 
			&=  \sum_{n=0}^{N} \Eb \Bigg( \omega_n \Bigg(  H_n  - \sum_{i=1}^{n} \rho_i \Delta S_i \prod_{j=i+1}^{n} \left(1 - \beta_j \Delta S_j \right)   -    c   \prod_{i=1}^{n} \left(1 - \beta_i \Delta S_i \right) \Bigg)^2 \Bigg) \\
			&= c^2 \, \Eb \Bigg(\sum_{n=0}^N \omega_n  \prod_{i=1}^{n} \left(1 - \beta_i \Delta S_i \right)^2 \Bigg) 
			- 2c \, \sum_{n=0}^N \Eb \Bigg( \omega_n H_n \prod_{i=1}^{n} \left(1 - \beta_i \Delta S_i \right) \Bigg) \\
			&\quad + \sum_{n=0}^{N} \Eb \Bigg( \omega_n \Bigg(  H_n  - \sum_{i=1}^{n} \rho_i\Delta S_i \prod_{j=i+1}^{n} \left(1 - \beta_j \Delta S_j \right)    \Bigg)^2 \Bigg) + 2c \cdot \mathbb{CT}, \qquad \label{eq:cal_V}
		\end{align} 
	}
	where the Cross-Term $\mathbb{CT} $ is defined by
	\begin{align}
	\mathbb{CT} :=  \sum_{n=0}^{N} \Eb \left( \omega_n \left(  \sum_{i=1}^{n} \rho_i\Delta S_i \prod_{j=i+1}^{n} \left(1 - \beta_j \Delta S_j \right)    \right) \,   \prod_{k=1}^{n} \left(1 - \beta_k \Delta S_k \right) \right).
	\end{align}
	
	Using \eqref{eq:A_eq}, we obtain that
	\begin{align}
		\label{eq:Z_pos}
		\Eb \Bigg(\sum_{n=0}^N \omega_n  \prod_{i=1}^{n} \left(1 - \beta_i \Delta S_i \right)^2 \Bigg) 
		= \Eb[A_0] = \sum_{n=0}^N \Eb[Z_n] \ge 0,
	\end{align}
	where $Z_n$ is defined in \eqref{eq:Z}.  Also by \eqref{eq:Z}, the second term in \eqref{eq:cal_V} becomes $2c \sum_{n=0}^N \Eb(H_n Z_n)$. 
	By comparing with \eqref{eq:value}, we see that Assertion (b) is proved if $\mathbb{CT} = 0$, which is done in the sequel:
	%
	{\small 
		\begin{align}
			\mathbb{CT} 
			&= \Eb \Bigg[\sum_{i=1}^N   \prod_{k=1}^{i} \left(1 - \beta_k \Delta S_k \right) \, 
			\Eb \Bigg[ \sum_{n=i}^N \omega_n \rho_i\Delta S_i  \prod_{j=i+1}^{n} \left(1 - \beta_j \Delta S_j \right) 
			\cdot (1 - \beta_i \Delta S_i)  \prod_{k=j+1}^{n} \left(1 - \beta_k \Delta S_k \right) \Bigg| \Fc_{i-1}\Bigg] \Bigg] \\
			&= \Eb \Bigg[\sum_{i=1}^N   \prod_{k=1}^{i} \left(1 - \beta_k \Delta S_k \right) 
			\cdot \Eb \Bigg[ \sum_{n=i}^N \omega_n \rho_i\Delta S_i  \prod_{j=i+1}^{n} \left(1 - \beta_j \Delta S_j \right)^2  (1 - \beta_i \Delta S_i)   \Bigg| \Fc_{i-1}\Bigg] \Bigg] \\
			&= \Eb \Bigg[\sum_{i=1}^N   \prod_{k=1}^{i} \left(1 - \beta_k \Delta S_k \right) \rho_i \, 
			\underbrace{\Eb \Bigg[ \Delta S_i  \sum_{n=i}^N \omega_n  \prod_{j=i+1}^{n} \left(1 - \beta_j \Delta S_j \right)^2   \Bigg| \Fc_{i-1}\Bigg]}_{ = \alpha_i \text{ by } \eqref{eq:alpha} \text{ and } \eqref{eq:A_eq}} \Bigg] \\
			&\quad - \Eb \Bigg[\sum_{i=1}^N   \prod_{k=1}^{i} \left(1 - \beta_k \Delta S_k \right) \rho_i \beta_i \, \cdot 
			\underbrace{\Eb \Bigg[ \Delta S_i^2  \sum_{n=i}^N \omega_n  \prod_{j=i+1}^{n} \left(1 - \beta_j \Delta S_j \right)^2   \Bigg| \Fc_{i-1}\Bigg]}_{ = \delta_i \text{ by } \eqref{eq:delta}} \Bigg] \\
			&= \Eb \Bigg[\sum_{i=1}^N   \prod_{k=1}^{i} \left(1 - \beta_k \Delta S_k \right) \cdot \rho_i \underbrace{(\alpha_i - \beta_i \delta_i)}_{=0 \text{ by } \eqref{eq:beta}} \Bigg] = 0, 
	\end{align}}
	\normalsize
	where in the first line we have used $\prod_{k=1}^n = \prod_{k=1}^i \, \prod_{k=i}^i \, \prod_{k=j+1}^n$ to simplify the computations on conditional expectation.

	The proof to the main theorem, Theorem \ref{maintheorem1}, is now complete.
\end{proof}

\section*{Acknowledgments}
We would like to thank Area Editor Professor Agostino Capponi, an anonymous Associate Editor, and an anonymous referee for their helpful comments and critics on early versions of this paper.
The research of Jun Deng is supported by the National Natural Science Foundation of China (11501105).
Bin Zou acknowledges support through a start-up grant from the University of Connecticut.

\section*{Supplementary Material}

Please refer to the online supplementary material for an example of Problem \eqref{eq:ran_prob} and the complete proofs to Lemma \ref{propintegrable1}, Corollary \ref{corowelldefined}, and Lemma \ref{lem:rho_beta}.

\bibliographystyle{elsarticle-num} 
\bibliography{ref}

\appendix

\section{Extension to Multiple Risky Assets}
\label{sec_multiple}
In the main context of the paper, we consider  only one risky  asset $S$. 
In this section, we generalized the main results in Theorem \ref{maintheorem1} to the case of multiple risky assets.

In the remaining of this section, introduce an $M$-dimensional process $S =(\vec{S}_n)_{n \in \Tc}$ to represent the $M$ risky assets in the market, where $\vec S_n = (S_{n1}, S_{n2},\cdots, S_{nM})^\top$, with $\ ^\top$ denoting the usual transpose operator. 
A trading strategy is now denoted by $\xi = (\vec \xi_n)_{n \in \Tc^+}$, where $\vec \xi_n = (\xi_{n1}, \xi_{n2},\cdots, \xi_{nM})^\top$ and $\xi_{ij}$ is the number of shares invested in the $j$-th risky asset from time $i-1$ to time $i$.
Given a strategy $\xi$, we define the corresponding gain process $G(\xi)=(G_n(\xi))_{n \in \Tc}$ by
\begin{align}
	\label{eq:G_multi}
	G_n(\xi) = \sum_{i=1}^n \, \vec{\xi}_i^\top \, \Delta \vec S_i, \quad n \in \Tc^+, \quad \text{ and } \quad 
	G_0(\xi) = 0. \quad 
\end{align}

Similar to Section \ref{secmainresult}, we define the following $M$-dimensional predictable processes $\beta$ and $\rho$ by 
\begin{align}
	\vec \beta_n &:= \frac{\vec \alpha_n}{\vec \delta_n} \qquad \text{ and } \qquad   \vec \rho_n :=   \frac{\vec \eta_n}{\vec \delta_n},  \quad \text{ for  all } n \in \Tc^+, \quad 
	\label{eq:beta_multi}
\end{align}
where $\alpha = (\vec \alpha_n)_{n \in \Tc^+}$, $\eta = (\vec \eta_n)_{n \in \Tc^+}$, and $\delta = (\vec \delta_n)_{n \in \Tc^+}$ are given by  
\begin{align}
	\vec	\alpha_n &:= \Eb\bigg[  \Delta \vec S_n \bigg[\sum_{i=n}^{N} \omega_i \prod_{j=n+1}^{i} (1 - \vec \beta_j^\top \Delta \vec S_j)\bigg] \bigg\vert {\cal F}_{n-1}\bigg],  
	\\
	\vec \eta_n &:= \Eb\bigg[ \Delta \vec S_n \bigg[\sum_{i=n}^{N} H_i \omega_i \prod_{j=n+1}^{i} (1 - \vec \beta_j^\top \Delta \vec S_j)\bigg] \bigg\vert {\cal F}_{n-1}\bigg], \quad
	\label{eq:eta_multi}
	\\
	\vec	\delta_n &:= \Eb \bigg[ \Delta \vec S_n^\top \Delta \vec S_n \bigg[\sum_{i=n}^{N} \omega_i \prod_{j=n+1}^{i} (1 - \vec \beta_j^\top  \Delta \vec S_j)^2 \bigg] \bigg \vert {\cal F}_{n-1}\bigg].
	\label{eq:delta_multi}
\end{align}

We now present the multi-dimensional version of Theorem \ref{maintheorem1} as follows.

\begin{theorem}
	\label{maintheorem1_multi} 
	For any fixed initial capital $c $, we define $\xi^\star(c) : = \left(\vec \xi^\star_n(c)\right)_{n \in \Tc^+}$ by
	\begin{align}
		\label{eq:op_xi_multi}
		\vec \xi_n^\star(c) := \vec \rho_n - \vec \beta_n  \big( c + G_{n-1}(\vec \xi^\star (c) ) \big),  
	\end{align}
	where $\vec \beta_n$ and $\vec \rho_n$ are defined in \eqref{eq:beta_multi} and $G_n$ is defined in \eqref{eq:G_multi}.
	The following results hold true:
	\begin{itemize}
		\item[\rm (a)] $V(c) = \Jc(\xi^\star; c) < \infty$ for all $c \in \Rb$, and the strategy $\xi^\star(c)$ defined in \eqref{eq:op_xi_multi} solves Problem \eqref{eq:main_1}.
		
		\item[\rm (b)] The value function   is given by 
		\begin{align}
			\label{eq:value_multi}
			V(c) &= c^2  \, \sum_{n=0}^{N}  \Eb \left( Z_n \right) - 2c \sum_{n=0}^{N} \Eb \left(    Z_n H_n  \right) \\
			&\quad +
			\sum_{n=0}^{N} \Eb \left( \omega_n \left(  H_n  - \sum_{i=1}^{n} \vec \rho_i^\top \Delta \vec S_i \prod_{j=i+1}^{n} \left(1 - \vec \beta_j^\top \Delta \vec S_j \right)    \right)^2 \right), 
		\end{align}
		where $Z = (Z_n)_{n \in \Tc}$ is defined by 
		\begin{align}
			\label{eq:Z_multi}
			Z_n := \omega_n \, \prod_{i=1}^n \, (1 - \vec \beta_i^\top \Delta \vec S_i). 
		\end{align}
		
		\item[\rm (c)] Let $n$ be an arbitrary but fixed non-negative integer in $\Tc$, i.e., $n = 0,1,\cdots,N$.
		We have, for all $k = 0, 1, \cdots, n+1$, that  (setting $G_{-1} = 0$)
		\begin{align}
			\label{eq:H_thm_multi}		
			H_n - c  - G_n(\vec \xi^\star(c)) = H_n - \sum_{i=k}^n \vec \rho_i^\top \Delta \vec S_i \prod_{j=i+1}^{n} (1 - \vec \beta_j^\top \Delta \vec S_j) 
			 - \big(c + G_{k-1}(\vec \xi^\star(c)) \big) \prod_{i=k}^n (1 - \vec \beta_i^\top \Delta \vec S_i).
		\end{align}
	\end{itemize}
\end{theorem}

\clearpage
\newpage 

\pagestyle{plain}
\onecolumn

\setcounter{page}{1}

\setcounter{section}{0}
\renewcommand\thesection{\alph{section}}

\setcounter{equation}{1}
\renewcommand{\theequation}{\thesection.\arabic{equation}}
\numberwithin{equation}{section}

\begin{center}
Online Supplementary Material for ``Quadratic  Hedging  for  Sequential Claims with Random Weights in Discrete Time"\\[1ex]
Jun Deng\footnote{School of Banking and Finance,
	University of International Business and Economics, Beijing, China. Email: jundeng@uibe.edu.cn} and Bin Zou\footnote{Corresponding author.   Department of Mathematics, University of Connecticut, 341 Mansfield Road U1009, Storrs, Connecticut 06269-1009, USA. Email: bin.zou@uconn.edu. Phone: +1-860-486-3921.}
\\[1ex]
\today
\end{center}

\section{Extra Example}
\label{sub:ex2}

In this example, we study a quadratic hedging problem under random horizon as formulated in \eqref{eq:ran_prob}. 
Consider a two-period binomial model $(\Omega, \Fc, ({\cal F}_{n})_{n=0,1,2}, \Pb)$  specified as follows: 
\begin{itemize}
	\item $\Omega = \{x_1, x_2, x_3, x_4\}$, ${\cal F}_0=\{\emptyset, \Omega\}$, ${\cal F}_1 = \sigma(\{x_1, x_2\}$, $\{x_3, x_4\})$, and ${\cal F}_2 = {\cal F} = 2^{\Omega}$.
	
	\item  $\Pb(x_1) = p^2$, $\Pb(x_2) = \Pb(x_3) = p q$, and $\Pb(x_4) = q^2$, where $0<p<1$ and $q = 1 - p$.  
	
	\item The stock price $S$ evolves by: at time $0$, $S_0= 1$; at time $1$, $S_1(\{x_1, x_2\}) = u$ and $S_1(\{x_3, x_4\}) = d$; at time $2$, 
	$S_2(\{x_1\}) = u^2$,     $S_2(\{x_2\}) = S_2(\{x_3\}) = ud$, and $S_2(\{x_4\}) = d^2$, where $0<d<1<u$.
\end{itemize}
We set up the contingent claim $H$ by 
\begin{align}
	\label{eq:cc_ex2}
	H_0 =a_0, \quad H_1 = a_1 \In_{\{x_1,x_2\}} + a_2 \In_{\{x_3,x_4\}}, \quad H_2 = b_1 \In_{\{x_1\}} + b_2 \In_{\{x_2\}} + b_3 \In_{\{x_3\}} + b_4 \In_{\{x_4\}},
\end{align}
where all the $a_i$'s and $b_i$'s  are constants. The random time $\tau$ is defined by 
\begin{align}
	\label{eq:tau_ex2}
	\tau = 0 \cdot \In_{\{x_1,x_2\}} + 1 \cdot \In_{\{x_3\}} + 2 \cdot \In_{\{x_4\}}.
\end{align}
To use the results from Theorem \ref{maintheorem1}, we require $\omega_n = \Pb(\tau = n | \Fc_n)$ for all $n=0,1,2$, which yields $\omega_0 = p$, $\omega_1 = p \In_{\{x_3,x_4\}}$, and  $\omega_2 = \In_{\{x_4\}}$.

Similar to the previous example, we carry out calculations by \eqref{eq:beta}-\eqref{eq:delta} and obtain 
$\alpha_2 = d(d-1) q  \In_{\{x_3,x_4\}}$, $\eta_2 = d(d-1) b_4 q \In_{\{x_3,x_4\}}$,  and $\delta_2 = d^2(d-1)^2 q \In_{\{x_3,x_4\}}$, implying $\beta_2 = \frac{\In_{\{x_3,x_4\}}}{d(d-1)}$ and  $\rho_2 = \frac {b_4 \In_{\{x_3,x_4\}}}{d(d-1)}$. 
At time 1, we compute  
$\alpha_1 = (d-1) pq$,  $\eta_1 = (d-1)a_2 pq$,  $\delta_1 = (d-1)^2 pq$, leading to $\beta_1 = \frac{1}{d-1}$ and $\rho_1 = \frac{a_2}{d - 1}$.

The sequence $(Z_n)_{n=0,1,2}$, defined in \eqref{eq:Z}, reads in this example as $Z_0 = p $ and $Z_1 = Z_2 = 0$. 
In turn, we get $\sum_{n=0}^2 \Eb[Z_n] = p $ and $\sum_{n=0}^N \Eb[Z_n H_n] =  a_0 p$. 
An application of Theorems \ref{maintheorem1} and \ref{thm:pricing} yields:

\begin{corollary}
	\label{cor:ex2}
	Let $H$ and $\tau$ be defined by \eqref{eq:cc_ex2} and \eqref{eq:tau_ex2}, and $\omega_n = \Pb(\tau = n | \Fc_n)$ for $n=0,1,2$.  For any initial capital $c$, the optimal hedging strategy $\xi^\star = \xi^\star(c)$ to Problem \eqref{eq:main_1} is given by 
	\begin{align}
		\label{eq:op_ex2}
		\xi_1^\star(c)  = \frac{a_2 - c}{d-1} \qquad \mbox{and} \qquad  \xi_2 ^\star(c) =  \frac{b_4 - a_2}{d(d-1)} \In_{\{x_3,x_4\}}.
	\end{align}
	The minimum capital $c^\star$ to Problem \eqref{pricingproblem} is $c^\star = a_0$.
\end{corollary}

\begin{proposition}
	\label{prop:tau}
	Let the assumptions in Corollary \ref{cor:ex2} hold. We have: 
	
	\begin{itemize}
		\item[\rm (a)] The optimal strategy $\xi^\star(c^\star)$ with initial capital $c^\star=a_0$, where $\xi^\star$ is given by \eqref{eq:op_ex2}, replicates the contingent claim $H=(H_0, H_1, H_2)$ on $\Omega$, $\{x_3, x_4\}$, and $\{x_4\}$, respectively.
		
		\item[\rm (b)] The stopped market $S^\tau$ admits arbitrage.
	\end{itemize}
\end{proposition}

\begin{proof}
	Assertion (a) can be verified by using \eqref{eq:op_ex2} from Corollary \ref{cor:ex2}. 
	To show Assertion (b), take an admissible strategy $\phi=(\phi_1, \phi_2)$ with $\phi_1 = \phi_2 = -1$. Then, by \eqref{eq:G}, we obtain 
	\begin{align}
		G_1&=\phi_1 \Delta S^\tau_1 = (1-d) \In_{\{x_3, x_4\}}\geq 0 ,\\ 
		G_2&=\phi_1 \Delta S^\tau_1 + \phi_2 \Delta S^\tau_2 = (1-d) \In_{\{x_3, x_4\}} + d(1-d) \In_{\{x_4\}} \geq 0,
	\end{align}
	and $\Pb(G_2 > 0) = q > 0 $. Hence, $\phi$ is an arbitrage strategy, which proves Assertion (b).
\end{proof}


\section{Proof of  Lemma \ref{propintegrable1}}
\label{proof:ABCD}

\begin{proof}
	We first prove Assertion (2) by induction. When $n = N$, we get $A_N = D_N = \omega_N \in [0,1]$, so \eqref{AKDKleq} holds trivially. 
	
	Next, suppose \eqref{AKDKleq} is true for all $N, N-1, \cdots, n+1$, where $n < N$.  We need to show that \eqref{AKDKleq} is also true for $n$.  To show this, we compute 
	\begin{align}
		\Eb[D_n | \Fc_n] 
		&= \Eb \big[\omega_n + \sum_{i=n+1}^{N} \omega_i \sum_{j=n+1}^{i} (1 - \beta_j \Delta S_j)^2 \big| \Fc_n \big]  
		\tag*{\text{(by \eqref{eq:ABCD})}}\\
		&=\Eb \big[\omega_n + \left(1 - \beta_{n+1} \Delta S_{n+1}\right)^2 \cdot \Eb[D_{n+1}|\Fc_{n+1}]  \big| \Fc_n \big]  
		\tag*{\text{(by tower rule)}}\\
		&= \omega_n + \Eb \big[ (1- 2 \beta_{n+1}\Delta S_{n+1} ) \underbrace{\Eb[D_{n+1}|\Fc_{n+1}]}_{=\Eb[A_{n+1}|\Fc_{n+1}]} \big| \Fc_n \big]  
		\tag*{\text{(by assumption)}}\\
		&\quad +  \beta_{n+1}^2  \, \underbrace{\Eb \big[\Delta S_{n+1}^2 D_{n+1} | \Fc_{n+1} \big]}_{=\delta_{n+1}}  \tag*{\text{(by \eqref{eq:delta})}} \\
		&=\omega_n + \Eb[A_{n+1}|\Fc_n] - 2 \beta_{n+1} \alpha_{n+1} + \beta_{n+1}^2 \delta_{n+1}  \tag*{\text{(by \eqref{eq:eta})}}
		\\
		&=\omega_n + \Eb[A_{n+1}|\Fc_n] - \beta_{n+1} \alpha_{n+1} \tag*{\text{(by \eqref{eq:beta})}} \\
		&= \Eb[A_n | \Fc_n] .  \tag*{\text{(by \eqref{eq:A_eq})}}
	\end{align}
	Recall $\alpha_{n+1} \beta_{n+1} = \beta_{n+1}^2 \delta_{n+1}$ and $\delta_{n+1} \ge 0$, and $\Eb[A_{n+1} | \Fc_{n+1}] \le N - n$ by assumption, we then have
	\begin{align}
		\Eb[A_n | \Fc_n]  \le \omega_n + \Eb \big[ \Eb[A_{n+1} | \Fc_{n+1}] \big| \Fc_n \big] \le N - n + 1,
	\end{align}
	which, together with the above results, confirms \eqref{AKDKleq} holds for all $n \in \Tc$.
	
	Our next objective is to show Assertion (1). To that end, we deduce
	\begin{align}
		\Eb \left[ A_n^2  | \Fc_n \right] &= \Eb \left(\left(  \sum_{i=n}^{N} \omega_i \prod_{j=n+1}^{i} (1 - \beta_j \Delta S_j)\right)^2 \Bigg \vert {\cal F}_{n} \right)  \tag*{\text{(by \eqref{eq:ABCD})}}\\
		&\le  (N-n+1) \Eb \left(  \sum_{i=n}^{N} \omega_i \prod_{j=n+1}^{i} (1 - \beta_j \Delta S_j)^2 \Bigg \vert {\cal F}_{n} \right)
		\tag*{\text{(By Cauchy-Schwarz and $\omega_i \in [0,1]$)}} \\
		&= (N-n+1) \, \Eb[D_n | \Fc_n]  \tag*{\text{(by \eqref{eq:ABCD})}} \\
		&\le (N - n + 1)^2,  \tag*{\text{(by \eqref{AKDKleq})}}
	\end{align}
	which readily shows $A_n \in \Lc^2(\Pb)$ for all $n \in \Tc$. 
	Using this result, we immediately obtain the locally square integrability of $B$ by 
	\begin{align}
		\Eb[B_n^2 \vert {\cal F}_{n-1}]  = \Eb\big[ A_n^2 \Delta S_n^2 \vert {\cal F}_{n-1}\big] = \Eb \big[ \Delta S_n^2 \, \Eb[A_n^2 | \Fc_n] \vert {\cal F}_{n-1}\big] 
		\le (N - n + 1)^2 \, \Eb \left[\Delta S_n^2 \vert {\cal F}_{n-1} \right] < \infty, \quad \forall \, n \in \Tc,
	\end{align}
	where we have used the fact that $S \in \LPL$. 
	Lastly, to see $C$ is also square integrable, we obtain 
	\begin{align}
		\Eb \left[ C_n^2\right] &= \Eb \left[ \beta_n^2 \Delta S_n^2 A_n^2 \right] = \Eb \left[ \beta_n^2 \Delta S_n^2 \, \Eb[A_n^2 | \Fc_n]  \right]\\
		&\leq (N-n+1)\Eb \left[ \beta_n^2 \Delta S_n^2 \,  \Eb[D_n | \Fc_n]  \right]  \\
		&= (N - n + 1) \, \Eb \left[ \beta_n^2 \delta_n \right] =  (N - n + 1) \, \Eb  \left[\frac{\alpha_n^2}{\delta_n}\right] \tag*{\text{(By \eqref{AKDKleq} and \eqref{eq:delta})}} \\
		&\le (N - n + 1)^2.  \tag*{\text{(By Cauchy-Schwarz and $\omega_i \in [0,1]$)}}
	\end{align}
	The proof is now complete.
\end{proof}

\section{Proof of Corollary \ref{corowelldefined}}
\begin{proof}
	By Jensen's inequality and assertion (1) of   Lemma \ref{propintegrable1}, we deduce that
	\begin{align}
		\alpha_n^2 &\leq \left(\Eb \left[B_n \vert \Fc_{n-1} \right]\right)^2 
		\leq \Eb \left[B_n^2 \vert \Fc_{n-1} \right] <+\infty.
	\end{align}
	By utilizing assertion (2) of  Lemma \ref{propintegrable1} and the definition in \eqref{eq:delta}, we have
	\begin{align}
		\delta_n &= \Eb\left[ \Delta S_n^2 D_n \vert \Fc_{n-1}\right] 
		=\Eb\left[ \Delta S_n^2 \, \Eb [D_n\vert \Fc_{n}] \vert \Fc_{n-1}\right]  
		\leq (N-n+1)\, \Eb\left[ \Delta S_n^2\vert \Fc_{n-1}\right] <+\infty.
	\end{align}
	Using Cauchy-Schwarz inequality, $\omega_i \in [0,1]$, and the definition of $D$ in \eqref{eq:ABCD}, we obtain
	\begin{align}
		\vert \eta_n \vert &\leq \sum_{i=n}^{N} \Eb\bigg[\vert  \Delta S_n \vert \cdot  \bigg\vert  H_i \omega_i \prod_{j=n+1}^{i} (1 - \beta_j \Delta S_j)\bigg\vert\,  \bigg\vert {\cal F}_{n-1}\bigg]\\
		&\leq \sum_{i=n}^{N} \left(\Eb\bigg[ H_i^2 \bigg\vert {\cal F}_{n-1}\bigg]\right)^{1/2} \cdot   \left(\Eb\bigg[  \Delta S_n^2      \omega_i \prod_{j=n+1}^{i} (1 - \beta_j \Delta S_j)^2    \bigg\vert {\cal F}_{n-1}\bigg]\right)^{1/2}\\
		&\leq \sum_{i=n}^{N} \left(\Eb\bigg[ H_i^2 \bigg\vert {\cal F}_{n-1}\bigg]\right)^{1/2} \cdot 
		\left(\Eb\bigg[  \Delta S_n^2     D_n   \bigg\vert {\cal F}_{n-1}\bigg]\right)^{1/2}\\
		& =  \sum_{i=n}^{N} \left(\Eb\bigg[ H_i^2 \bigg\vert {\cal F}_{n-1}\bigg]\right)^{1/2} \cdot 
		\left(\delta_n\right)^{1/2} < +\infty,
	\end{align}
	where the last inequality is due to $H_n \in \LP$ for all $n$.
	The proof is now complete.
\end{proof}

\section{Proof of  Lemma \ref{lem:rho_beta}}
\label{proof:rho_beta}
\begin{proof}
	By definition \eqref{eq:beta}, we readily see $\beta_n \eta_n = \alpha_n \rho_n$, which reads as 
	\begin{align}
		\Eb\left[\beta_n \Delta S_n  \sum_{i=n}^{N} \omega_i H_i \prod_{j=n+1}^{i} (1 - \beta_j \Delta S_j) \Bigg\vert {\cal F}_{n-1}\right] = \Eb\left[\rho_n \Delta S_n \sum_{i=n}^{N} \omega_i   \prod_{j=n+1}^{i} (1 - \beta_j \Delta S_j) \Bigg\vert {\cal F}_{n-1}\right],
	\end{align}
	Using the above result, we derive
	\begin{align}
		\text{l.h.s.}
		&= \Eb \left[ \Delta S_n \left(\sum_{j=n+1}^{N}  \rho_j\Delta S_j  \left(  \sum_{i=j}^{N}  \omega_i \prod_{k=j+1}^{i} \left(1 - \beta_k \Delta S_k \right)\right) \right) \Bigg\vert {\cal F}_{n-1}\right] \\
		&= E\left[ \Delta S_n \left(\sum_{j=n+1}^{N}  \beta_j\Delta S_j  \left(  \sum_{i=j}^{N}  \omega_i H_i \prod_{k=j+1}^{i} \left(1 - \beta_k \Delta S_k \right)\right) \right) \Bigg\vert {\cal F}_{n-1}\right] 
		=\text{r.h.s.}
	\end{align}
	The proof is now complete.
\end{proof}

\end{document}